\documentclass[twoside,final,11pt,final]{entics} 

\usepackage{enticsmacro}
\usepackage{amsmath, amsfonts, amssymb}
\usepackage{mathtools, stmaryrd}
\usepackage{graphicx}
\usepackage[all]{xy}
\usepackage[english]{babel}
\usepackage{tikz}
\usepackage{tikz-cd}
\usepackage{pgfplots}
\usepackage{fancybox}
\usepackage{stmaryrd}
\usepackage{mathtools}
\usepackage{enumitem}

\usepackage{todonotes}

\def\twoheaddownarrow{\rlap{$\downarrow$}\raise-.5ex\hbox{$\downarrow$}}
\def\twoheaduparrow{\rlap{$\uparrow$}\raise.5ex\hbox{$\uparrow$}}

\newcommand{\join}{\ensuremath{\bigvee}} 

\newcommand{\upperSetOf}[1]{\ensuremath{\mathop{\uparrow} {#1}}} 

\newcommand{\po}{\mathbf{Po}\mathrm{}}    

\newcommand{\Nat}{{\mathbb N}}
\newcommand{\Real}{{\mathbb R}}

\def\conf{MFPS 2026} 	
\volume{6}			
\def\volu{6}			
\def\papno{2}			
\def\mydoi{19402}
\def\lastname{Bennett and Farjudian} 
\def\runtitle{Robustness Analysis via Horofunction Compactification} 
\def\copynames{H. Bennett, A. Farjudian}    
\begin{document}
\begin{frontmatter}
  \title{Robustness Analysis via Horofunction Compactification} 						
  \author{Harrison Bennett\thanksref{a}\thanksref{myemail}}	
   \author{Amin Farjudian\thanksref{a}\thanksref{coemail}}		
     \address[a]{School of Mathematics\\ University of Birmingham\\				
    Birmingham, United Kingdom}  							
   \thanks[myemail]{Email: \href{mailto:hxb496@student.bham.ac.uk} {\texttt{\normalshape
        hxb496@student.bham.ac.uk}}} 
  \thanks[coemail]{Email:  \href{mailto:A.Farjudian@bham.ac.uk} {\texttt{\normalshape        A.Farjudian@bham.ac.uk}}}
\begin{abstract} 
Robustness analysis plays a central role in the verification and design of computational and hybrid systems, particularly when system behaviour depends continuously on parameters subject to perturbation. Existing domain-theoretic frameworks provide a principled foundation for reasoning about such perturbations via monotone maps on lattices of closed sets. However, these frameworks face significant limitations when the underlying state space is not locally compact, as is the case for the infinite-dimensional spaces that arise in analysis, machine learning, and control theory (e.g., $\ell_p$ and $L_p$ spaces). In these settings, the lattice of closed subsets fails to be continuous, and classical compactifications either sacrifice precision or lack computable structure.

We propose Gromov's horofunction compactification as a new tool for robustness analysis over a class of separable metric spaces of practical importance, including separable reflexive Banach spaces. Given a metric space $\mathbb{S}$, we show that its horofunction extension yields a compact metric space together with a Lipschitz embedding, which enables robust approximations of monotone maps via Scott-continuous maps on the compactified domain. For separable spaces, the horofunction compactification is metrizable, which provides a path toward effective domain-theoretic constructions.
\end{abstract}
\end{frontmatter}

\section{Introduction}\label{intro}

We introduce a framework for robustness analysis over separable metric spaces. A system is said to be robust with respect to perturbations of a set of parameters if small changes to those parameters do not lead to significant changes in the behaviour of the system. Robustness analysis is a core topic in system analysis, including machine learning and cyber-physical systems.

We are particularly interested in spaces that are not locally compact. Examples of such spaces include sequence spaces $\ell_p$ and Lebesgue spaces $L_p(\mathbb{R}^n)$, for $p \in [1, \infty)$, and $n \geq 1$. These spaces are fundamental in functional analysis and related areas such as partial differential equations. The key step in our method is the metric compactification of a given state space via the so-called horofunctions~\cite{Gromov:Horofunctions:1981}.

We work within the topological framework established by Moggi et al.~\cite{moggi2018safe}. Assume that $\mathbb{S} := (S, \to)$ is a transition system, $P( S)$ denotes the powerset of $S$,  and $\rho_\mathbb{S}: P( S) \to  P( S)$ is a reachability map that returns, for any given set $ A \subseteq S$, the set of states reachable from $A$. Various types of reachability may be considered, e.g., reachability in finitely many transitions, reachability over the entire operation of the system, etc. As such, reachability maps are functions that are applied on \emph{subsets} of the state space rather than points of the space. 

Reachability maps are monotone, in the sense that $A \subseteq B \implies \rho_{\mathbb{S}}(A) \subseteq \rho_{\mathbb{S}}(B)$. Since we approximate subsets with outer approximations, we first consider working with the lattice $( P(S), \supseteq)$, i.e., with superset relation as the order of information because smaller sets are more accurate.

Although the lattice $( P(S), \supseteq)$ is continuous, it is $\omega$-continuous (i.e., has a countable basis) if and only if $S$ is (at most) countable. We are interested in state spaces such as Banach spaces, which are uncountable. As such, $P(S)$ is too large. However, in reality, physical measurements have finite precision, and mathematical models have finite fidelity. In the presence of such imprecisions, it is (in general) impossible to distinguish a subset $A$ of a metric space from its closure $\overline{A}$~\cite[Section~2]{moggi2019system}. Hence, we consider the smaller lattice $(\mathbb{C}(\mathbb{S}), \supseteq)$ of closed subsets of the state space. 

Moggi et al. introduced the \emph{robust topology} on the lattice $(\mathbb{C}(\mathbb{S}), \supseteq)$ of closed subsets of $S$~\cite[Definition~A.1]{moggi2018safe}, which captures robustness, in the sense that, a function $f: ({\mathbb{C}}(\mathbb{S}_1)), \supseteq) \to (\mathbb{C}(\mathbb{S}_2), \supseteq)$ is robust with respect to small perturbations of its input if and only if it is continuous with respect to the robust topologies on $(\mathbb{C}(\mathbb{S}_1), \supseteq)$ and $(\mathbb{C}(\mathbb{S}_2) , \supseteq)$~\cite[Theorem~A.2]{moggi2018safe}. In general, on $(\mathbb{C}(\mathbb{S}), \supseteq)$, the robust topology is finer than the Scott topology. When $\mathbb{S}$ is a compact metric space, however, the two topologies coincide~\cite[Theorem~A.4]{moggi2018safe} and $(\mathbb{C}(\mathbb{S}), \supseteq)$ is an $\omega$-continuous lattice. This is useful because, over effectively given domains, Scott continuity is a necessary condition for computability~\cite{Smyth:effectively_given_domains:1977} and if $(\mathbb{C}(\mathbb{S}), \supseteq)$ is $\omega$-continuous, then one has the basic ingredients for computing robust approximations of monotone maps, e.g., a robust approximation of the reachability map $\rho_{\mathbb{S}}$ for a transition system $\mathbb{S} := (S, \to)$.

If $\mathbb{S}$ is not compact, then the robust topology may be strictly finer than the Scott topology~\cite[Example~A.5]{moggi2018safe}. If, in addition, $\mathbb{S}$ is not locally compact, then $(\mathbb{C}(\mathbb{S}), \supseteq)$ is not a continuous lattice. In such cases, one may look for a substitute $\omega$-continuous lattice $\mathbb{L} := (L, \sqsubseteq)$ which is related to $(\mathbb{C}(\mathbb{S}), \supseteq)$ via an adjunction:
\begin{equation}
\label{eq:basic_Galois}
\begin{tikzcd}[column sep = large]
    \mathbb{C}(\mathbb{S}) \arrow[r, "\top"', "\iota_*" , yshift =
    1.1ex] & \mathbb{L}  \arrow[l, dashed, "\iota^*" , yshift = -1.1ex]
\end{tikzcd}     
\end{equation}
in the category of complete lattices and monotonic maps. In~\cite{farjudian2023robustness}, a construction is proposed which, for a broad class of metric spaces $\mathbb{S}$, generates a compact Hausdorff space $\hat{\mathbb{S}}$ for which $\mathbb{L} = \mathbb{C}(\hat{\mathbb{S}})$ is an $\omega$-continuous lattice.

The construction of~\cite{farjudian2023robustness} provides an explicit description of the compact space $\hat{\mathbb{S}}$ and a countable basis for the $\omega$-continuous lattice $\mathbb{C}(\hat{\mathbb{S}})$. When $\mathbb{S}$ is a subset of a Euclidean space, the space $\hat{\mathbb{S}}$ is a compactification of $\mathbb{S}$. When $\mathbb{S}$ is infinite-dimensional, however, the space $\hat{\mathbb{S}}$ is not a compactification of $\mathbb{S}$. For instance, if $\mathbb{S}$ is the closed unit ball of $\ell_p$ for $p \in [1, \infty]$, the space $\hat{\mathbb{S}}$ is the closed unit ball of $\ell_p$ with the weak-* topology~\cite[Theorem~5.12]{farjudian2023robustness}. Since the weak-* topology is coarser than the norm topology, this can lead to a major loss of precision, as detailed in~\cite[Section~5.2]{farjudian2023robustness}.

\subsection{Contributions} 
In this paper, we present a method that provides a metric compactification $\hat{\mathbb{S}} = (\hat{S}, d_{\hat{\mathbb{S}}})$ for a class of separable metric spaces $\mathbb{S} = (S, d_{\mathbb{S}})$ of practical importance. We use the \emph{horofunction compactification} introduced by Gromov~\cite{Gromov:Horofunctions:1981}. A key feature of this compactification is that the embedding $h: \mathbb{S} \to \hat{\mathbb{S}}$ is Lipschitz continuous (Proposition~\ref{prop_embedding_lipschitz}). As a result, the right adjoint $\iota_*: \mathbb{C}(\mathbb{S}) \to \mathbb{C}(\hat{\mathbb{S}})$ from the adjunction in~\eqref{eq:basic_Galois} is a robust map (Theorem~\ref{thm:right_adj_robust}).

For the special case of reflexive $\ell_p$ spaces (i.e., $p \in (1, \infty)$), we present an explicit description of a countable basis for the $\omega$-continuous lattice $(\mathbb{C}(\hat{\mathbb{S}}), \supseteq)$ of closed subsets  of the horofunction compactification, which is a key ingredient in developing a computational framework via effectively given domains~\cite{Smyth:effectively_given_domains:1977}. We also investigate loss of precision in our framework via an example of a non-convex closed subset of $\ell_p$.

\subsection{Related Work}

As mentioned earlier, we work within the framework established by Moggi et al.~\cite{moggi2018safe}, which is based on various concepts from topology, domain theory, and category theory. In~\cite{moggi2018safe}, it was established that, for any metric space $\mathbb{S} = ( S, d)$, the robust topology on $\mathbb{C}(\mathbb{S})$ includes the Scott topology, and the two coincide when $\mathbb{S}$ is compact. Further work, including some preliminary results involving non-compact spaces, was presented in~\cite{moggi2019system,Moggi_Farjudian_Taha:System_Analysis_and_Robustness:ICTCS:2019}. A general construction yielding an $\omega$-continuous lattice $\mathbb{L}$ as in the adjunction~\eqref{eq:basic_Galois} was introduced in~\cite{farjudian2023robustness}. The framework was further extended to  generalised (quantale-valued) metric spaces in~\cite{Dagnino_Farjudian_Moggi:Robustness_Quantales_Hausdorff_Smyth:ICTAC:2023,Dagnino_Farjudian_Moggi:Robust_Topology:arXiv:2025}. The framework of~\cite{moggi2018safe} was applied in introducing a domain-theoretic framework for robustness analysis of neural networks in~\cite{Zhou_Shaikh_Li_Farjudian:Robust_NN:MSCS:2023}.

The idea of using substitute domain-theoretic constructions for dealing with non-locally-compact spaces has also been applied in solutions of ordinary differential equations via abstract bases in~\cite{Edalat_Farjudian_Li:Temporal_Discretization:2023}. It was shown in~\cite{Farjudian_Jung:Spectral:ENTICS:2024} that the same construction can be obtained via the so-called spectral compactification.

There are well-known methods of compactification in the literature, such as the one-point (Alexandroff) compactification which is suitable only for locally-compact spaces, and Stone-{\v C}ech compactification which is suitable for Tychonoff (in particular, metric) spaces~\cite{munkres2017topology}, but its existence requires the axiom of choice, and as such, it is not suitable for a constructive framework. 

The Stone-{\v C}ech compactification can be very large and complex, but there are ways of obtaining smaller and constructive compactifications for separable metric spaces. One such method, which we call the \emph{horofunction compactification}, was proposed by Gromov~\cite{Gromov:Horofunctions:1981}. Our work relies on the results of~\cite{gutierrez2019metric,daniilidis2025horofunction}.

\subsection{Structure of the Paper}

The remainder of the paper is organised as follows:

\begin{itemize}
    \item Section~\ref{sec:prelim} recalls the mathematical background required for our development. We review key notions from domain theory, including continuous dcpos and Scott topology, and summarise the robust topology of Moggi et al.~\cite{moggi2018safe}, which serves as the conceptual foundation for robustness analysis over metric spaces. We also recall relevant topological preliminaries, such as compactifications and pointwise convergence.

    \item Section~\ref{sec:Adjunctions} develops the adjunction-based framework that underlies our approach. We show how any compactification $\hat{\mathbb{S}}$ of a metric space $\mathbb{S}$ yields an adjunction between the lattices $\mathbb{C}( \mathbb{S})$ and $\mathbb{C}(\hat{\mathbb{S}})$, and we analyse the topological behaviour of the induced adjoint maps. A central result of this section establishes that whenever the embedding $\iota: \mathbb{S} \to \hat{\mathbb{S}}$ is Lipschitz, one can construct robust approximations in a systematic way (Corollary~\ref{cor:compos_approx}).

    \item Section~\ref{sec:horofunction_compactification} introduces Gromov’s horofunction compactification and develops the corresponding horofunction extension $h: \mathbb{S} \to \overline{\mathbb{S}}^h$. We provide a self-contained treatment of the relevant theory and prove that the resulting compactification is metrizable and admits a $2$-Lipschitz embedding. This establishes the suitability of the horofunction compactification as a robust-preserving construction.

    \item   Section~\ref{sec:case_ell_p} specialises the framework to the case $\mathbb{S} = \ell_p$ with $1 < p < \infty$. We give an explicit construction of a countable basis for the $\omega$-continuous lattice $\mathbb{C}(\overline{\ell_p}^h)$, which is needed for an effective domain-theoretic framework~\cite{Smyth:effectively_given_domains:1977}. We also show, via Example~\ref{example:ell_p_precision}, that the horofunction approach can retain more information than the previously proposed construction of~\cite{farjudian2023robustness}.

    \item   Section~\ref{sec:concluding_remarks} concludes with a discussion of limitations and future directions, including effective representations for $\mathbb{C}(\overline{\mathbb{S}}^h)$ and the extension of our techniques to $L_p(X)$ spaces.    
\end{itemize}

\section{Mathematical Preliminaries}
\label{sec:prelim}

In this article, we write $A \subseteq_f B$ to denote that $A$ is a finite subset of $B$.

\subsection{Domain Theory}

We recall some basic concepts from domain theory~\cite{AbramskyJung94-DT,Goubault-Larrecq:Non_Hausdorff_topology:2013}. Assume that $\mathbb{P} := ( P, \sqsubseteq)$ is a partially ordered set (poset). A subset $A \subseteq P$ is said to be directed if it is non-empty and every finite subset has an upper bound in $A$, i.e.:
\begin{equation*}
    A \neq \emptyset \text{ and } \forall x, y \in A, \exists z \in A: \ (x \sqsubseteq z) \wedge (y \sqsubseteq z).
\end{equation*}
We write $A \subseteq_{\mathrm{dir}} P$ to indicate that $A$ is a directed subset of $P$.

A directed-complete partially ordered set (dcpo) is a poset $\mathbb{P}$ that is closed under joins of directed subsets. Intuitively, each element of a directed set $A \subseteq P$ provides some partial information about the join $\join{A}$. As such, dcpos provide a primitive structure for a computation framework. The full structure is provided by \emph{continuous} dcpos, also known as (continuous) domains.

Assume that $\mathbb{D} = ( D, \sqsubseteq)$ is a dcpo. We say that an element $x \in D$ is way-below $y \in D$ (written $x \ll y$) if $\forall A \subseteq_{\mathrm{dir}} D: (y \sqsubseteq \join A \implies \exists a \in A: x \sqsubseteq a)$. For each $x \in D$, we define \textcolor{black}{$\twoheaddownarrow x := \{ z \in D \mid z \ll x\}$.} The way-below relation is also known as the \emph{approximation} relation, and in domain-theoretic terms, the approximants of $x$ are the elements of $\twoheaddownarrow x$. A subset $B \subseteq D$ is said to be a \emph{basis} for $\mathbb{D}$ if every element $x \in D$ is the join of its approximants in the set $B$, i.e., if we let $B_x := B \cap \twoheaddownarrow x$, then for every $x \in D$, the set $B_x$ is directed and $x = \join B_x$. A dcpo which has a basis is said to be \emph{continuous}. In this article, by a domain we mean a continuous dcpo with a bottom element $\bot_D$.

For any subset $A \subseteq P$ of a poset $( P, \sqsubseteq)$, we define $\upperSetOf{A} := \{ x \in P \mid \exists a \in A: a \sqsubseteq x \}$. A subset $O \subseteq D$ of a dcpo is said to be Scott open if:
\begin{enumerate}
    \item $O$ is an upper set, i.e., $O = \upperSetOf{O}$.
    \item $O$ is inaccessible by directed joins, i.e., $\forall A \subseteq_{\mathrm{dir}} D: \join A \in O \implies A \cap O \neq \emptyset$.
\end{enumerate}
The collection of such subsets forms a topology on $D$ called the Scott topology. It turns out that the structure preserving maps on dcpos are exactly those that are continuous with respect to the Scott topology, i.e., for any dcpos $(D_1, \sqsubseteq_1)$ and  $(D_2, \sqsubseteq_2)$ and any function $f: D_1 \to D_2$, $f$ is Scott-continuous if and only if it preserves joins of directed sets~\cite[Proposition~2.3.4]{AbramskyJung94-DT}.

\subsection{Robustness}

We recall the definition of a robust map from~\cite[Deﬁnition~4.1]{moggi2018safe}. Given a metric space $(S,d)$ and a subset $A \subseteq S$, let $A_\delta$ denote the closure of the open set $B(A,\delta) := \{y \in S \mid \exists x \in A : d(x,y)<\delta\}$.
For any two metric spaces $\mathbb{S}_1$ and $\mathbb{S}_2$, a monotonic map $f:(\mathbb{C}(\mathbb{S}_1), \supseteq) \to (\mathbb{C}(\mathbb{S}_2), \supseteq)$ is said to be \emph{robust} at $C \in \mathbb{C}(\mathbb{S}_1)$ if:
\begin{equation}
\label{eq:robustness_epsilon_delta}
\forall \epsilon > 0, \exists \delta > 0: \ f(C_\delta)\subseteq f(C)_\epsilon   
\end{equation}
where $(\mathbb{C}(\mathbb{S}), \supseteq)$ denotes the lattice of closed subsets of $\mathbb{S}$ ordered with superset relation. The map $f$ is said to be robust if it is robust at every $C \in \mathbb{C}(\mathbb{S}_1)$.

Given a metric space $S$, a subset $U \subseteq \mathbb{C}(\mathbb{S})$ is said to be \emph{robust open} if
\begin{equation*}
    \forall C \in U, \exists \delta>0: \ \upperSetOf{B(C,\delta)} \subseteq U.
\end{equation*}
The collection of all such sets forms the so-called \emph{robust topology} on $\mathbb{C}(\mathbb{S})$. This topology indeed captures robustness in the sense that:

\begin{theorem}[{\cite[Theorem~A.2]{moggi2018safe}}]
\label{theorem_monotonic_robust}
Given a map $f:\mathbb{C}(\mathbb{S}_1)\to \mathbb{C}(\mathbb{S}_2),$ for metric spaces $\mathbb{S}_1$ and $\mathbb{S}_2$, the following properties are equivalent:
\begin{enumerate}
\item $f$ is a monotonic and robust map.
\item $f$ is continuous with respect to the robust topologies on $\mathbb{C}(\mathbb{S}_1)$ and $\mathbb{C}(\mathbb{S}_2)$.
\end{enumerate}
\end{theorem}

For any metric space $\mathbb{S} := (S, d)$, the Scott topology on $\mathbb{C}(\mathbb{S})$ is included in the robust topology~\cite[Lemma~A.3]{moggi2018safe}. When $\mathbb{S}$ is also compact, the Scott and robust topologies on $\mathbb{C}(\mathbb{S})$ coincide~\cite[Theorem~A.4]{moggi2018safe}. As a result, when $( S_1, d_1)$ and $( S_2, d_2)$ are two compact metric spaces, a map $f: \mathbb{C}(\mathbb{S}_1) \to \mathbb{C}(\mathbb{S}_2)$ is robust if and only if it is Scott-continuous. We also point out that robustness becomes trivial on discrete spaces, that is, if $( S_1, d_1)$ is discrete and $( S_2, d_2)$ is an arbitrary metric space, then every monotonic map $f: \mathbb{C}(\mathbb{S}_1) \to \mathbb{C}(\mathbb{S}_2)$ is robust~\cite[Proposition~1]{moggi2019system}.

\subsection{Topology}

By a compactification of a topological space $\mathbb{X} := (X, \tau)$ we mean a compact space $\hat{\mathbb{X}} := ( \hat{X}, \hat{\tau})$ together with a dense embedding $\iota : \mathbb{X} \hookrightarrow \hat{\mathbb{X}}$, i.e., $\overline{\iota(X)}^{\hat{\mathbb{X}}} = \hat{X}$, in which $\overline{\iota(X)}^{\hat{\mathbb{X}}}$ denotes the topological closure of $\iota(X)$ in $\hat{\mathbb{X}}$.

\begin{definition}
\label{def_pointwise}
    Let $X$ be a set and $\mathbb{Y} = (Y,\tau)$ be a topological space. For a point $x \in X$ and an open set $U \subseteq Y$, let 
    \begin{equation*}
      S(x,U) := \{f: X \to Y \mid f(x)\in U\}.  
    \end{equation*}
    The collection of all sets $S(x,U)$ is a subbasis for a topology on $\mathbb{Y}^X$ which is called the \emph{topology of pointwise convergence}.
\end{definition}

Note that the topology of pointwise convergence coincides with the product topology. We also recall the definition of equicontinuity from~\cite[Section~45]{munkres2017topology}. Let $\mathbb{X} := (X,d)$ and $ \mathbb{Y} := (Y, d')$ be metric spaces, and let $\mathcal{F}$ be a subset of the set $C(\mathbb{X}; 
\mathbb{Y})$ of continuous functions from $\mathbb{X}$ to $\mathbb{Y}$. If $x_0 \in X$, the set $\mathcal{F}$ of functions is said to be \emph{equicontinuous at $x_0$}
if, given $\varepsilon > 0$, there exists a neighbourhood $U_\epsilon$ of $x_0$ such that:
\begin{equation*}
\forall f \in \mathcal{F}, \forall x \in U_\epsilon: \ d'\bigl(f(x), f(x_0)\bigr) < \varepsilon.    
\end{equation*}
If the set $\mathcal{F}$ is equicontinuous at $x_0$ for each $x_0 \in X$, it is said 
to be \emph{equicontinuous}.

\section{Adjunctions and a Robust-Scott Correspondence}
\label{sec:Adjunctions}

In this article, $\po$ denotes the category of complete lattices and monotonic maps. An \emph{adjunction} in $\po$ is a pair of morphisms $f:X \rightarrow Y$ and $g:Y \rightarrow X$ such that $f \circ g \le \text{id}_Y$ and $ \text{id}_X\le g \circ f$. It is written as $f \dashv g$ and the maps $f$ and $g$ are called the left and right adjoints, respectively. For adjunctions in  \textbf{Po}, we have the following characterisation which is also called a \emph{Galois connection}: $f \dashv g$ if and only if:
\begin{equation*}
 \forall x \in X, \forall y \in Y:\ x \leq_X g(y) \iff f(x) \leq_Y y.
\end{equation*}

Let $\mathbb{S} := (S, d)$ be a metric space and let $\hat{\mathbb{S}} := ( \hat{S}, \hat{d})$ be a compactification of it with embedding $\iota: S \to \hat{S}$. We define an adjunction $\iota^* \dashv \iota_*$ between the lattices of closed sets $\mathbb{C}(\mathbb{S})$ and $\mathbb{C}(\hat{\mathbb{S}})$. The aim is to construct the maps $\iota^*$ and $\iota_*$ so that given a map $f: \mathbb{C}(\mathbb{S}) \rightarrow \Sigma$, with $\Sigma$ being the two point (Sierpi{\'n}ski) lattice, and a Scott-continuous approximation $\hat{f}$ of $f \circ \iota^*$, the map $\hat{f} \circ \iota_*$ will be a robust approximation of $f$ (Figure~\ref{fig:compactification_diagram}). Note that $\Sigma$ is also of the form $\mathbb{C}(\mathbb{T})$, in which $\mathbb{T} := (\{ * \}, d_\mathbb{T})$ is the compact metric space of a singleton with the discrete metric $d_\mathbb{T}$.
\begin{figure}[h] 
    \centering
    \Large 
    \begin{tikzcd}
        {\mathbb{C}(\mathbb{S})} & \top & {\mathbb{C}(\hat{\mathbb{S}})} \\
        & \Sigma \cong \mathbb{C}(\{ * \})
        \arrow["{\iota_*}", shift left=3, from=1-1, to=1-3]
        \arrow["f"', from=1-1, to=2-2]
        \arrow["{  \iota^*}", shift left=2, from=1-3, to=1-1]
        \arrow["{\hat{f}}", from=1-3, to=2-2]
    \end{tikzcd}
    \caption{A Diagram showing the relationship between the adjunctions $\iota_*$, $\iota^*$, and the maps $f$, $\hat{f}$.}
    \label{fig:compactification_diagram}
\end{figure}
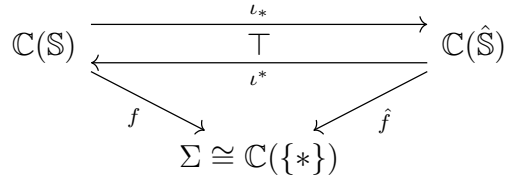
\begin{remark}
We have considered maps into the Sierpi{\'n}ski space $\Sigma = \mathbb{C}(\mathbb{\{ * \}})$ rather than into a more general space $\mathbb{C}(\mathbb{K})$, for some compact metric space $\mathbb{K}$, partly for simplicity, and partly because the so-called \emph{safety analyses} are indeed maps into the Sierpi{\'n}ski space~\cite[Section~1]{farjudian2023robustness}. In this respect, we work within the same setting as~\cite{farjudian2023robustness}.
\end{remark}

We define the candidate right and left adjoints for the adjunction as follows:
 \begin{align*}
\forall A \in \mathbb{C}(\mathbb{S}): \quad      &\iota_* (A) := \overline{\iota(A)}^{\hat{\mathbb{S}}},\\
          \forall B \in \mathbb{C}(\hat{\mathbb{S}}): \quad &\iota^* (B) :=\iota^{-1}(B).
 \end{align*}
Note that $\iota^*(B)= \{x\in S \mid  \iota(x)\in B \} = \iota^{-1}(B \cap \iota(S) )$. It is straightforward to verify that $\iota_*$ and $\iota^*$ are well-defined monotonic maps between $\mathbb{C}(\mathbb{S})$ and $\mathbb{C}(\hat{\mathbb{S}})$.

\begin{proposition}
    The maps $\iota_*$ and $\iota^*$ form an adjunction.
\end{proposition}

\begin{proof}
    Let $A \in \mathbb{C}(\mathbb{S})$ and $B \in \mathbb{C}(\hat{\mathbb{S}})$. To show that the maps $\iota_*$ and $i^*$ form an adjunction, we must show that the following Galois connection is satisfied:
    \begin{equation*}
      B \leq \iota_*(A) \iff \iota^*(B) \leq A,
    \end{equation*}
    where $\leq$ is the superset relation.  

\begin{description}
    \item[$(\impliedby)$] By assumption, we have: $\iota^*(B) \leq A \iff A \subseteq \iota^{-1}(B \cap \iota(S))$. This means that $\iota(A)$ is a subset of $B \cap \iota(S)$ which implies $\iota(A) \subseteq B.$ Since $B$ is closed in $\mathbb{\hat{S}}$, it follows that $\overline{\iota(A)}^{\mathbb{\hat{S}}} \subseteq B $ and hence $B \leq \iota_*(A)$.

    \item[$(\implies)$] By assumption, $\overline{\iota(A)}^{\mathbb{\hat{S}}} \subseteq B$, so $\iota(A) \subseteq B.$ Since $A \subseteq S$ it must also follow that $\iota(A) \subseteq B \cap \iota(S)$. By definition of the preimage this means $A \subseteq \iota^{-1}(\iota(A)) \subseteq \iota^{-1}(B \cap \iota(S))$  and hence $\iota^*(B) \leq A.$
\end{description} 
 \end{proof}

For any $A \in \mathbb{C}(\mathbb{S})$, we have $A \subseteq \iota^*(\iota_*(A))$. When we obtain equality, i.e., $A = \iota^*(\iota_*(A))$, we say that there is \emph{no loss of precision}. In cases of inequality, the difference between $A$ and $\iota^*(\iota_*(A))$ provides a measure of loss of precision. A detailed analysis of precision was presented in~\cite[Section~5]{farjudian2023robustness} for the construction introduced in~\cite{farjudian2023robustness}. In Example~\ref{example:ell_p_precision}, we will show that the horofunction construction can retain more precision when compared to the construction of~\cite{farjudian2023robustness}.

\subsection{Metric Compatibility}

In favourable cases, the embedding $\iota : \mathbb{S} \to \hat{\mathbb{S}}$ can be an isometry, e.g., when $\mathbb{S} := ( 0, 1)$ and $\hat{\mathbb{S}} := [0,1]$. In this section, we demonstrate that, for some spaces of interest (e.g., infinite-dimensional Banach spaces) no embedding associated with the compactification can be an isometry (Corollary~\ref{cor:iota_not_isometry}).

Recall that $\mathbb{S} := (S, d)$ is said to be \emph{bounded} if $\exists K \in \mathbb{R},  \forall x,y \in S: d(x,y) \leq K$. It is said to be \emph{totally bounded} if for all $\epsilon > 0$ there exists a finite collection of open balls of radius $\epsilon$  whose union contains $S$. Every compact metric space is totally bounded.

\begin{proposition}
\label{prop:totally_bdd_bdd}
    If $\mathbb{S} := ( S, d)$ is totally bounded, then it is bounded.
\end{proposition}
\begin{proof}
The proof is straightforward, see, e.g.,\cite[Example~1, p.~273]{munkres2017topology}
\end{proof}

\begin{proposition}
\label{prop:subset_total}
  Subspaces of totally bounded sets are totally bounded.
\end{proposition}

\begin{proof}
Suppose that $\mathbb{S}$ is totally bounded with $A \subseteq S$. Fix $\epsilon >0$ and let $a_1, \ldots, a_n$ be the centres of an $\epsilon/2$ covering of $\mathbb{S}$. 
Without loss of generality, assume that $a_1, \ldots, a_k$ are the centres of those balls that intersect $A$, hence, $\forall i  \in \{1, \dots, k \}: b_{\epsilon /2}(a_i) \cap A \ne \emptyset.$ For each $i \in \{1, \dots, k\}$, choose some point $b_i \in b_{\epsilon /2}(a_i) \cap A$. We claim that $\{b_{\epsilon}(b_i) \}_{i = 1}^k$ is a covering of $A$. To show this, observe that:
   \begin{equation*}
       \forall b \in A, \exists j \in \{1, \dots, k \}: \quad b \in b_{\epsilon/2}(a_j).
   \end{equation*}
   By the triangle inequality, we have $d(b_j, b) \le d(b_j, a_j) + d(a_j, b)  < \epsilon.$
\end{proof}

\begin{proposition}
\label{prop:imb_totally}
   Let $\mathbb{S} = (S, d_{\mathbb{S}})$ be a metric space with compactification $\hat{\mathbb{S}} = (\hat{S}, d_{\hat{\mathbb{S}}})$ and associated embedding $\iota: \mathbb{S} \hookrightarrow \hat{\mathbb{S}}$. If $\iota$ is isometric, then $\mathbb{S}$ is totally bounded.
\end{proposition}

\begin{proof}
    Let $\iota: \mathbb{S} \hookrightarrow \hat{\mathbb{S}}$ be an isometric embedding. Since $\hat{\mathbb{S}}$ is compact (and therefore totally bounded), from Proposition~\ref{prop:subset_total}, the image of $S$ under $\iota$ is also totally bounded. This means
    \begin{equation*}
    \forall \epsilon > 0, \exists p_1, \ldots, p_n \in \iota(S): \       \iota(S) \subseteq \bigcup_{i =1}^{n} b^{\hat{\mathbb{S}}}_{\epsilon}(p_i).
    \end{equation*}
    For each $p_i,$ let $q_i = \iota^{-1}(p_i) \in S$. Using the fact that $\iota$ is an isometry and taking the inverse of both sides we have $S \subseteq \bigcup_{i =1}^{n} \iota^{-1}(b^{\hat{\mathbb{S}}}_{\epsilon}(p_i))$. By definition,
    \begin{align*}
    \iota^{-1}(b^{\hat{S}}_{\epsilon}(p_i)) &= \{x \in S \mid \iota(x) \in b^{\hat{\mathbb{S}}}_{\epsilon}(p_i) \} \\
    &= \{x \in S \mid d_{\hat{\mathbb{S}}}(\iota(x), \iota(q_i)) < \epsilon \} \\
   (\text{$\iota$ is an isometry}) &= \{x \in S \mid d_{\mathbb{S}}(x,q_i) < \epsilon \} \quad 
    \\ &= b_\epsilon^{\mathbb{S}}(q_i).
\end{align*}
It follows that $S \subseteq \bigcup_{i=1}^{n}b_\epsilon^\mathbb{S}(q_i)$. Hence, $\mathbb{S}$ is totally bounded.
\end{proof}

We have shown that for any compactification, if we require the embedding $\iota: \mathbb{S} \hookrightarrow \hat{\mathbb{S}}$ to be isometric, then $\mathbb{S}$  must be totally bounded and consequently bounded (Proposition~\ref{prop:totally_bdd_bdd}).
Also, by~\cite[Theorem~45.1]{munkres2017topology}:
\begin{equation}
\label{eq:compact_tot_bdd_complete}
 \text{} \mathbb{S} \text{ is compact} \iff \mathbb{S} \text{ is totally bounded + complete.}   
\end{equation}

\begin{corollary}
\label{cor:iota_not_isometry}
    If $\mathbb{S}$ is a non-compact and complete metric space, then the embedding $\iota: \mathbb{S} \hookrightarrow  \hat{\mathbb{S}}$ cannot be an isometry.
\end{corollary}

\begin{proof}
By Proposition~\ref{prop:imb_totally}, if $\iota$ is an isometry, then $\mathbb{S}$ must be totally bounded. If, furthermore, $\mathbb{S}$ is complete, then by~\eqref{eq:compact_tot_bdd_complete}, it must be compact, which is a contradiction.
\end{proof}

As a consequence, for the following spaces, the embedding $\iota$ into the compactification cannot be isometric:

\begin{itemize}
    \item Metric spaces that are not totally bounded such as the following (when equipped with the usual norm and metric topology):  $\mathbb{R}$, $\mathbb{R}^n$, $\ell_p$ spaces, and $L_p(X)$ ($1 \le p \leq \infty$), for a measurable space  $(X, \Sigma, \mu)$.
    \item Non-compact, complete metric spaces, including closed unit balls of infinite-dimensional Banach spaces such as $\ell_p$ and $L_p(X)$.
\end{itemize}

Nevertheless, it turns out that, for our purposes, Lipschitz continuity of the embedding 
$\iota$ is sufficient (Theorem~\ref{thm:right_adj_robust}). This is useful because the embedding associated with our horofunction compactification is indeed $2$-Lipschitz (Proposition~\ref{prop_embedding_lipschitz}).

\subsection{Adjoint Maps and Their Topological Properties }
\label{subsec:ajoint_top_prop}

The left adjoint $\iota^*: \mathbb{C}(\hat{\mathbb{S}}) \to \mathbb{C}(\mathbb{S})$ is known to be Scott-continuous~\cite[Proposition~3.1.14]{AbramskyJung94-DT}. The right adjoint $\iota_*: \mathbb{C}(\mathbb{S}) \to \mathbb{C}(\hat{\mathbb{S}})$, on the other hand, is Scott-continuous if and only if $\mathbb{C}(\mathbb{S})$ is a continuous lattice~\cite[Theorem~3.1.4.]{AbramskyJung94-DT}, which is, in turn, equivalent to $\mathbb{S}$ being locally compact~\cite[Theorem~5.2.9]{Goubault-Larrecq:Non_Hausdorff_topology:2013}.\footnote{To be precise, requiring $\mathbb{C}(\mathbb{S})$ to be a continuous lattice is equivalent to $\mathbb{S}$ being \emph{core-compact}. But since $S$ is assumed to be metrizable (hence, sober), it is equivalent to requiring local compactness.} We are interested in non-locally-compact spaces such as infinite-dimensional Banach spaces.

Although the left adjoint $\iota^*: \mathbb{C}(\hat{\mathbb{S}}) \to \mathbb{C}(\mathbb{S})$ is Scott-continuous, it may not be continuous with respect to the robust topology on $\mathbb{C}(\mathbb{S})$.

\begin{example}
\label{example:left_adjoint_not_robust}
    Assume that $\mathbb{S} = \mathbb{R}$ and $\hat{\mathbb{S}} = [ -1, 1]$ with the embedding $\iota = \tanh$ and the left adjoint $\iota^* : \mathbb{C}([ -1, 1]) \to \mathbb{C}(\mathbb{R})$. Then, the set $\{ \emptyset \} \subseteq \mathbb{C}(\mathbb{R})$ is a robust open subset of $\mathbb{C}(\mathbb{R})$, but
    \begin{equation*}
        (\iota^*)^{-1}(\{ \emptyset \}) = \{ \emptyset, \{ -1 \}, \{ 1 \}, \{ -1, 1\}\},
    \end{equation*}
 is not a Scott open subset of $\mathbb{C}([-1,1])$. For instance, for each $n \in \Nat$, define $A_n := [1 - 2^{-n}, 1] \in \mathbb{C}([ -1, 1])$. Then, $\join_{n \in \Nat} A_n = \{ 1 \} \in (\iota^*)^{-1}(\{ \emptyset \})$ but $\forall n \in \Nat: A_n \not \subseteq (\iota^*)^{-1}(\{ \emptyset \})$. Therefore, $\iota^* : \mathbb{C}([ -1, 1]) \to \mathbb{C}(\mathbb{R})$ is not continuous with respect to the robust topology on $\mathbb{C}(\mathbb{R})$ and the Scott topology on $\mathbb{C}([ -1, 1])$.
\end{example}

 In Theorem~\ref{thm:right_adj_robust}, we will show that, if $\hat{\mathbb{S}}$ is a metric space and $\iota : \mathbb{S} \to \hat{\mathbb{S}}$ is Lipschitz, then the right adjoint is robust. Note that the left adjoint still may not be robust. For instance, in Example~\ref{example:left_adjoint_not_robust}, the embedding ($\tanh$) is $1$-Lipschitz but the left adjoint is not continuous with respect to the robust topology on $\mathbb{C}([ -1, 1])$ either because $[-1,1]$ is a compact metric space and as a result, the robust and Scott topologies coincide on $\mathbb{C}([ -1, 1])$~\cite[Theorem~A.4]{moggi2018safe}.

Nevertheless, robust continuity of the right adjoint $\iota_*$ is sufficient for us to prove that the map $\hat{f} \circ \iota_*$ is a robust approximation of $f: \mathbb{C}(\mathbb{S}) \to \Sigma$ (Corollary~\ref{cor:compos_approx}). In Proposition~\ref{prop_embedding_lipschitz}, we will prove that the embedding associated with our horofunction compactification is $2$-Lipschitz.

\begin{lemma}
\label{lemma:right_adj_robust}
If the embedding $\iota: \mathbb{S} \hookrightarrow \hat{\mathbb{S}}$ is Lipschitz continuous, then the induced right adjoint
\begin{align*}
    \iota_*: \mathbb{C}(\mathbb{S}) \rightarrow \mathbb{C}(\hat{\mathbb{S}}), \quad \iota_*(A) = \overline{\iota(A)}^{\hat{\mathbb{S}}},
\end{align*}
is robust.
\end{lemma}

\begin{proof}
We use the $\epsilon$-$\delta$ formulation of~\eqref{eq:robustness_epsilon_delta}, that is, we prove that, for any given $A \in \mathbb{C}( \mathbb{S})$: 
\begin{equation*}
\forall \epsilon > 0 , \exists \delta > 0: \ 
\iota_*(A_\delta) \subseteq (\iota_*(A))_{\epsilon}.    
\end{equation*}
Let $k$ be the Lipschitz constant corresponding to the embedding $\iota$ and define $K := \max \{ k, 1\}$. Let $d_{\mathbb{S}}$ and $d_{\hat{\mathbb{S}}}$ be the metrics on $\mathbb{S}$ and $\hat{\mathbb{S}}$, respectively, then:
\begin{equation}
\label{eq:K_Lip_ineq}
    \forall x, y \in \mathbb{S}: \ d_{\hat{\mathbb{S}}}( \iota(x), \iota(y)) \leq K \cdot d_{\mathbb{S}}( x, y).
\end{equation}

We now state two facts which we will use later on: $\forall A \in \mathbb{P}(\mathbb{S}) , \forall \delta > 0$:
\begin{itemize}
    \item (F.1) $B(A, \delta) = B(\bar{A},\delta).$ \cite[Remark~4.2]{moggi2018safe}

    \item (F.2) $B(A,\delta) \subseteq A_\delta := \overline{B(A,\delta)} \subseteq B(A,\delta'),$ $ \forall  \delta' > \delta$. \cite[Proposition~2.6]{farjudian2023robustness}
\end{itemize}
\noindent
Additionally, we will introduce the following notation: for $A \in  \mathbb{C}(\mathbb{S})$ and $B \in \mathbb{C} (\hat{\mathbb{S}})$, write
\begin{align*}
    &B_\mathbb{S}(A,\delta) \coloneqq \{x \in \mathbb{S} \ | \ \exists a \in A, d_\mathbb{S}(x,a) < \delta \}, \\
    &B_{\hat{\mathbb{S}}}(B,\epsilon) \coloneqq  \{y \in \hat{\mathbb{S}} \ |\ \exists b \in B, d_{\hat{\mathbb{S}}}(y,b) < \epsilon \}.
\end{align*}

Let $A \in \mathbb{C}(\mathbb{S})$ and  $\epsilon > 0$ be given. Choose $\delta_0 < \epsilon/K$ and let $y \in \iota(B_\mathbb{S}(A,\delta_0))$. Then $y = \iota(x)$ for some $x \in B_\mathbb{S}(A,\delta_0).$ By definition of $B_\mathbb{S}(A,\delta_0)$, there exists an  $a \in A$ such that $d_{\mathbb{S}}(x,a) < \delta_0$. From~\eqref{eq:K_Lip_ineq}, we obtain:
\begin{equation*}
d_{\hat{\mathbb{S}}}(y,\iota(a)) =  d_{\hat{\mathbb{S}}}(\iota(x), \iota(a)) \leq  K \cdot d_{\mathbb{S}}(x,a) <K \delta_0 < \epsilon.  
\end{equation*}
Therefore, $\iota(a)$ witnesses the fact that $y \in B_{\hat{\mathbb{S}}}(\iota(A),\epsilon)$. Hence, we obtain:
\begin{equation*}
 \iota(B_\mathbb{S}(A,\delta_0)) \subseteq B_{\hat{\mathbb{S}}}(\iota(A) , K \delta_0)  \subseteq  B_{\hat{\mathbb{S}}}(\iota(A) , \epsilon).   
\end{equation*}
If we apply (F.2) in the space $\hat{\mathbb{S}}$ we get
\begin{equation*}
\forall \gamma >0: \quad
 \overline{\iota(B_\mathbb{S}(A,\delta_0)}^{\hat{\mathbb{S}}} \subseteq \overline{B_{\hat{\mathbb{S}}}(\iota(A) , K \delta_0)}^{\hat{\mathbb{S}}} \subseteq B_{\hat{\mathbb{S}}}(\iota(A), K\delta_0 + \gamma).
\end{equation*}
Since $K\delta_0 < \epsilon$, for any $\gamma < \epsilon - K\delta_0$, we have:
$$B_{\hat{\mathbb{S}}}(\iota(A), K\delta_0 + \gamma) \subseteq B_{\hat{\mathbb{S}}}(\iota(A), \epsilon).$$
\noindent
From (F.1), we know that $B_{\hat{\mathbb{S}}}(\iota(A), \epsilon) =  B_{\hat{\mathbb{S}}}( \overline{\iota(A)}^{\hat{\mathbb{S}}}, \epsilon)$, so if we put everything together, we get:
\begin{equation}
\label{eq:thm_intermediate}
 \overline{\iota(B_\mathbb{S}(A,\delta_0))}^{\hat{\mathbb{S}}} \subseteq B_{\hat{\mathbb{S}}}( \overline{\iota(A)}^{\hat{\mathbb{S}}}, \epsilon).   
\end{equation}

Now, choose any $\delta < \delta_0$, e.g., $\delta = \delta_0 / 2$. We have:
\begin{eqnarray}
    (\text{By F.2}) & & A_\delta \subseteq B_{\mathbb{S}}( A, \delta_0) \nonumber \\ 
    (\text{$\iota$ is monotonic}) & \implies & \iota(A_\delta) \subseteq \iota( B_{\mathbb{S}}( A, \delta_0)) \nonumber\\
    (\text{closure is monotonic})
    & \implies & \overline{\iota(A_\delta)}^{\hat{\mathbb{S}}}  \subseteq \overline{\iota( B_{\mathbb{S}}( A, \delta_0))}^{\hat{\mathbb{S}}}. \label{eq:closure_iota}
\end{eqnarray}
Finally, we obtain:
\begin{eqnarray*}
& & \iota_*( A_\delta)\\
(\text{By definition}) & = & \overline{\iota(A_\delta)}^{\hat{\mathbb{S}}} \\
(\text{By~\eqref{eq:closure_iota}}) & \subseteq & \overline{\iota( B_{\mathbb{S}}( A, \delta_0))}^{\hat{\mathbb{S}}}\\
(\text{By~\eqref{eq:thm_intermediate}}) & \subseteq & B_{\hat{\mathbb{S}}}( \overline{\iota(A)}^{\hat{\mathbb{S}}}, \epsilon)\\
(\text{By F.2}) & \subseteq & (\iota_*(A))_{\epsilon}.
\end{eqnarray*}
\end{proof}

\begin{theorem}
\label{thm:right_adj_robust}
If the embedding $\iota: \mathbb{S} \hookrightarrow \hat{\mathbb{S}}$ is Lipschitz continuous, then the right adjoint $\iota_*$ is continuous with respect to the robust topologies on $\mathbb{C}(\mathbb{S})$ and $\mathbb{C}(\hat{\mathbb{S}})$.
\end{theorem}

\begin{proof}
By Theorem~\ref{theorem_monotonic_robust}, the claim follows from Lemma~\ref{lemma:right_adj_robust} and the fact that the right adjoint is monotonic.
\end{proof}

\begin{corollary}
\label{cor:compos_approx}
If the embedding $\iota: \mathbb{S} \hookrightarrow \hat{\mathbb{S}}$ is Lipschitz continuous, then given a map $f: \mathbb{C}(\mathbb{S}) \to \Sigma$ and a Scott-continuous approximation $\hat{f} : \mathbb{C}(\hat{\mathbb{S}})  \to \Sigma$ of $f \circ \iota^*$, the map $\hat{f} \circ \iota_*: \mathbb{C}(\mathbb{S}) \rightarrow \Sigma$ is a robust approximation of $f$.     
\end{corollary} 

\begin{proof}
The metric spaces $\hat{\mathbb{S}}$ and (the singleton metric space) $\{ * \}$ are both compact. Hence, the Scott and robust topologies coincide on $\mathbb{C}(\hat{\mathbb{S}})$ and $\Sigma = \mathbb{C}(\{ * \})$, and as a result, $\hat{f}: \mathbb{C}(\hat{\mathbb{S}})  \to \Sigma$ is robust continuous. By Theorem~\ref{thm:right_adj_robust}, the right adjoint $\iota_*$ is also robust continuous. Therefore, the composition $\hat{f} \circ \iota_*: \mathbb{C}(\mathbb{S}) \rightarrow \Sigma$ is robust continuous.

Also, by assumption, $\hat{f}$ is an approximation of $f \circ \iota^*$. Hence:
\begin{eqnarray*}
& &    \hat{f} \sqsubseteq f \circ \iota^*\\
(\text{By monotonicity of function composition}) & \implies & \hat{f} \circ \iota_* \sqsubseteq f \circ \iota^* \circ \iota_* \\
(\iota^* \circ \iota_* \sqsubseteq \mathrm{id}_{\mathbb{C}(\mathbb{S})}) & \implies & \hat{f} \circ \iota_* \sqsubseteq f.\\
\end{eqnarray*}
As a result, $\hat{f} \circ \iota_*$ is indeed an approximation of $f$.
\end{proof}

\section{Horofunction Compactification}
\label{sec:horofunction_compactification}

In this section, we present a compactification for metric spaces called the \emph{horofunction compactification}, which was introduced by Gromov~\cite{Gromov:Horofunctions:1981}. We let  $\overline{\mathbb{S}}^h$ denote the horofunction compactification of a given metric space $\mathbb{S}$. 

For any metric space $\mathbb{S}$, the space $\overline{\mathbb{S}}^h$ is compact. In general, however, $\overline{\mathbb{S}}^h$ is not a compactification of $\mathbb{S}$ since the associated \emph{horofunction extension} $h : \mathbb{S} \to \overline{\mathbb{S}}^h$ is not always a topological embedding (see~\cite{gaubert2012maximin,daniilidis2025horofunction} for further details). When $\overline{\mathbb{S}}^h$ is a compactification of $\mathbb{S}$, we say that $\mathbb{S}$ is \emph{Gromov-compactifiable}, examples of which include reflexive Banach spaces~\cite[Corollary~3.8]{daniilidis2025horofunction}. See~\cite{daniilidis2025horofunction} for more examples.

\subsection{The Horofunction Extension}
\label{subsection_desript}

As described in~\cite[Section 1.2]{daniilidis2025horofunction}, for a metric space $\mathbb{S} := (S,d)$ and a fixed basepoint $x_0 \in S$, we assign a map $h_{x_0,z}$ to each $z \in S$ by:
\begin{equation*}
h_{x_0,z}(\cdot) \coloneqq d(\cdot,z) - d(x_0,z).   
\end{equation*}
Note that $h_{x_0}: \mathbb{S} \to C( \mathbb{S}; \mathbb{R})$, in which $C( \mathbb{S}; \mathbb{R})$ denotes the space of continuous real-valued functions on $\mathbb{S}$. Since the basepoint is fixed, for our purposes we remove $x_0$ from the subscript and simply write $h$ and $h_z$ instead of $h_{x_0}$ and $h_{x_0, z}$.

Given the above mapping, $\overline{\mathbb{S}}^h$ can be specified in different (but equivalent) ways~\cite{daniilidis2025horofunction}, e.g.:
\begin{enumerate}
    \item as the closure of $h(S)$ in $C(\mathbb{S}; \Real)$ with respect to the compact-open topology on $C(\mathbb{S}; \Real)$, 
    \item as the pointwise closure of $h( S)$
in the closed subspace $\mathrm{Lip}^1_{x_0}(\mathbb{S})$ of the product topology on $\mathbb{R}^\mathbb{S}$, in which 
$\mathrm{Lip}^1_{x_0}(\mathbb{S})$ is the space of all $1$-Lipschitz real-valued functions on $\mathbb{S}$ vanishing
at $x_0$:
\begin{equation*}
    \mathrm{Lip}^1_{x_0}(\mathbb{S}) \coloneqq  \{f: \mathbb{S}\rightarrow \mathbb{R} \mid f(x_0) = 0, \forall x,y \in S: |f(x) - f(y)| \leq d( x, y) \}.
\end{equation*}
\end{enumerate}
We focus on the second construction (see~\cite[Section~2.1]{daniilidis2025horofunction} for a full overview). Hence, $\overline{\mathbb{S}}^h = \overline{\{h_z \mid \textcolor{black}{z\in S\}}}$, where the closure is taken with respect to the topology of pointwise convergence.\footnote{The set $h(S) =\{h_{z} \mid z \in S\}$ is commonly referred to as the collection of \emph{internal metric functionals}.}

\begin{lemma}
\label{lemma_horo_lip}
For each basepoint $x_0 \in S$ and fixed $y \in S$, we have $h_y \in \mathrm{Lip}^1_{x_0}(\mathbb{S})$.
\end{lemma}

\begin{proof}
First, note that $h_y(x_0) = d(x_0,y) - d(y,x_0) = 0$. To show that $h_y$ is $1$-Lipschitz, assume that $x, z \in S$. We have:
\begin{equation*}
|h_y(x) - h_y(z)| = |d(x,y)-d(y,x_0) -d(z,y)+d(y,x_0)| = |d(x,y) -d(z,y)|.
\end{equation*}
By the triangle inequality, we obtain $|d(x,y) -d(z,y)| \leq d(x, z)$. Hence $|h_y(x) - h_y(z)| \leq d(x, z)$.
\end{proof}

\begin{lemma}
\label{lemma_horo_inj}
The map $h: \mathbb{S} \to \mathrm{Lip}^1_{x_0}(\mathbb{S})$ is injective and continuous.
\end{lemma}

\begin{proof} 
To prove injectivity, let $y, y' \in S$ and suppose that $h_y = h_{y'}$. Then: $\forall z \in S: d(z,y) - d(x_0,y) = d(z,y') - d(x_0,y')$. If we let $z=y$, then
\begin{equation*}
    d(y,y) -d(x_0,y) = d(y,y') - d(x_0,y')   \implies d(y,y') = d(x_0,y') - d(x_0,y).
\end{equation*} 
Now if we instead let $z = y'$, we get
\begin{equation*}
 d(y',y) - d(x_0,y) = d(y',y') - d(x_0,y') \implies d(y,y') =   d(x_0,y) - d(x_0,y').    
\end{equation*}
Adding the two equations gives us $d( y, y') = 0$, which implies $y = y'$. Therefore, the map $h(y) = h_y$ is injective.

To prove continuity, note that the space $\mathrm{Lip}^1_{x_0}(\mathbb{S})$ carries the subspace topology on $\mathbb{R}^\mathbb{S}$ equipped with the topology of pointwise convergence (equivalently, the product topology). Therefore, the map $h: \mathbb{S} \rightarrow \mathrm{Lip}^1_{x_0}(\mathbb{S})$ is continuous if and only if each coordinate map is continuous \cite[Theorem~19.6]{munkres2017topology}. For a fixed $x \in \mathbb{S}$, take the coordinate projection $\pi_x : \mathbb{R}^\mathbb{S} \rightarrow \mathbb{R}$ defined by:
\begin{equation*}
\forall f \in \mathbb{R}^\mathbb{S}: \quad  
 \pi_x(f) = f(x).   
\end{equation*}
The composite $ \pi_x \circ h: \mathbb{S} \rightarrow \mathbb{R}$ is given by, 
\begin{equation*}
(\pi_x \circ h)(z) = h_z(x) = d(x,z) - d(x_0,z).    
\end{equation*}
\noindent
This is clearly continuous in $z$. Therefore, for each fixed $x \in \mathbb{S}$, the map $\pi_x \circ h$  is continuous. Since each coordinate projection is continuous, the map $h$ is continuous as a map into $\mathbb{R}^\mathbb{S}$ and, hence also into $\mathrm{Lip}^1_{x_0}(\mathbb{S}).$
\end{proof}

The following variant of Ascoli's Theorem 
 is useful in demonstrating the equivalence of the two ways of obtaining the compactification:

\begin{theorem}[Ascoli's Theorem~{\cite[Theorem~47.1]{munkres2017topology}}] 
\label{theorem_ascoli}
Let $\mathbb{X}$ be a topological space and let $\mathbb{Y} := (Y, d)$ be a metric space.
Equip $ C(\mathbb{X}; \mathbb{Y})$ with the topology of compact convergence\footnote{Equivalently, this is the compact-open topology since $Y$ is a metric space.} and let $\mathcal{F}$ be a subset of $\mathcal{C}( \mathbb{X}; \mathbb{Y})$.
\begin{itemize}
\item If $\mathcal{F}$ is equicontinuous and for all $a \in X$ the set $\mathcal{F}_a := \{ f (a) \mid f \in \mathcal{F} \}$ has compact closure, then $\mathcal{F}$ is contained in a compact subspace of $\textcolor{black}{C} ( \mathbb{X}; \mathbb{Y})$. 
\item The converse holds if $\mathbb{X}$ is locally compact and Hausdorff.
\end{itemize}
\end{theorem}

For a given metric space $\mathbb{S} := ( S,d)$, we endow $C(\mathbb{S};\mathbb{R})$ with the compact-open topology. For a fixed basepoint $x_0 \in S$, we consider $\mathcal{F} = \mathrm{Lip}^1_{x_0}(\mathbb{S})$. For any $f \in \mathrm{Lip}^1_{x_0}(\mathbb{S})$ and $x \in S$, we have $|f(x)| \leq d(x, x_0)$ and hence $f(x) \in [-d(x,x_0), d(x,x_0)]$. For each $x \in S$, define $C_x := [-d(x,x_0), d(x,x_0)]$ and $\mathcal{F}_x = \{f(x) \mid f \in \mathrm{Lip}^1_{x_0}(\mathbb{S})\}$. As such, $\forall x \in S: \mathcal{F}_x \subseteq C_x$. Since $C_x$ is compact, $\mathcal{F}_x$ has compact closure for each $x \in S$. It is also straightforward to show that $\mathrm{Lip}^1_{x_0}(\mathbb{S})$ is equicontinuous. Furthermore:
\begin{equation*}
\mathrm{Lip}^1_{x_0}(\mathbb{S}) \subset \prod_{x \in \mathbb{S}}C_x .
\end{equation*}

Given an equicontinuous family $\mathcal{F}$, if $\mathcal{G}$ is defined to be the closure of $\mathcal{F}$ in the product topology on $\mathbb{Y}^X$, then the product topology on $\mathbb{Y}^X$ and the compact convergence topology on $\textcolor{black}{C}(\mathbb{X};  \mathbb{Y})$ coincide on the subset $\mathcal{G}$~\cite[Proof of Theorem~47.1]{munkres2017topology}. It is known that $\mathrm{Lip}^1_{x_0}(\mathbb{S})$ is indeed closed in the product topology. Hence, the product topology and the compact convergence topology coincide on $\mathrm{Lip}^1_{x_0}(\mathbb{S})$. Furthermore, when $\mathbb{Y}$ is a metric space, the compact open topology on $\textcolor{black}{C} (\mathbb{X}; \mathbb{Y})$  coincides with the topology of compact convergence and hence also the product topology~\cite[Theorem~46.8]{munkres2017topology}. From Lemmas~\ref{lemma_horo_lip} and ~\ref{lemma_horo_inj}, we know that $h$ is a continuous injection into $\mathrm{Lip}^1_{x_0}(\mathbb{S})$. If $\mathbb{S}$ is Gromov-compactifiable, its compactification $\overline{\mathbb{S}}^h$ is then obtained as the pointwise closure of $h(S)$ taken in $\mathrm{Lip}^1_{x_0}(\mathbb{S})$. 

As shown in~\cite[Section~2.1]{daniilidis2025horofunction}, $\overline{\mathbb{S}}^h$ is independent of the choice of the basepoint $x_0 \in S$ and furthermore:
\begin{proposition}[{\cite[Proposition 1.2]{daniilidis2025horofunction}}]
\label{prop:horo_dense_separable}
    Let $\mathbb{S} := ( S, d)$ be a metric space:
    \begin{enumerate}
        \item \label{item:dense_subset} If $Z$ is a dense subset of $S$, then $\overline{Z}^h = \overline{\mathbb{S}}^h$. 
        \item \label{item:separable_metrizable}  If $\mathbb{S}$ is separable,  $\overline{\mathbb{S}}^h$ is metrizable.
    \end{enumerate}
\end{proposition}

\subsection{Applying the Horofunction Extension.}

The product space $\mathbb{R}^{\mathbb{S}}$ is only metrizable if the indexing set $\mathbb{S}$ is countable. Non-trivial Banach spaces are uncountable. So, we require separability to ensure that the compactification is metrizable.
Given a separable metric space $\mathbb{S}$ with a countable dense subset $Z$, by Proposition~\ref{prop:horo_dense_separable}, one can restrict the domain of the embedding to $Z$ to obtain the map $h|_Z: \mathbb{S} \rightarrow \mathbb{R}^Z.$ Taking the closure of the image $h|_Z(\mathbb{S})$ in the pointwise topology on $\mathbb{R}^Z$ yields a space that is both metrizable and homeomorphic to the full horofunction compactification $\overline{\mathbb{S}}^h$. A metric which induces the product topology on $\mathbb{R}^Z$ is given by:
\begin{equation}
\label{eq:d_*}
\forall \mathbf{x},\mathbf{y} \in \mathbb{R}^Z  : \quad 
d_{*}(\mathbf{x},\mathbf{y})
=
\sum_{n=1}^{\infty} 2^{-n} \min(1, |x_n - y_n|).  
\end{equation}

\begin{proposition}\label{prop_embedding_lipschitz}

Assume that $(S,d)$ is a separable metric space and $Z \subseteq S$ is a countable dense subset. Equip $\mathbb{R}^Z$ with the metric $d_*$ from~\eqref{eq:d_*}. The horofunction extension $h: (S,d)\hookrightarrow (\mathbb{R}^Z,d_*)$ is Lipschitz continuous with Lipschitz constant $
    2$.
\end{proposition}

\begin{proof}
Assume that $Z = \{z_n\}_{n \in \mathbb{N}} \subseteq S$ and let $x,y \in S$. Define $ \mathbf{x} := (x_n)_{n \in \Nat}$ and $\mathbf{y} := (y_n)_{n \in \Nat}$ in $\mathbb{R}^Z$ by:
\begin{equation*}
\forall n \in \Nat: \ 
x_n := h_x(z_n), \quad y_n := h_y(z_n).
\end{equation*}
For any $n \in \mathbb{N}$, we have:
\begin{eqnarray*}
 |x_n - y_n| & = & |h_x(z_n) -  h_y(z_n)|\\
 & = &  \left| [d(z_n,x) - d(x_0,x)] - [d(z_n,y) - d(x_0,y)] \right|, \\
 & = & | [d(z_n,x) - d(z_n,y)] - [d(x_0,x) - d(x_0,y)] |, \\
 & \leq & |d(z_n,x) - d(z_n,y)| + |d(x_0,x) - d(x_0,y)| \\
 (\text{By the triangle inequality}) & \leq & d(x,y) + d(x,y) \\
 & =  & 2\cdot d(x,y).
\end{eqnarray*}
By the fact that $\sum_{n=1}^\infty 2^{-n} = 1$, we obtain $d_*(\mathbf{x}, \mathbf{y}) \leq 2 \cdot d(x,y)$.
\end{proof}

\begin{corollary}
   The induced map $h_* : \mathbb{C}(\mathbb{S}) \rightarrow \mathbb{C}(\overline{\mathbb{S}}^h)$ given by $h_*(A) = \overline{h(A)}^{\overline{\mathbb{S}}^h}$ is robust. 
\end{corollary}

\begin{proof}
    The proof follows from Proposition~\ref{prop_embedding_lipschitz} and Theorem~\ref{thm:right_adj_robust}.
\end{proof}

\section{The Case of $\mathbb{S} = \ell_p$}
\label{sec:case_ell_p}

In this section, we focus on the case of $\mathbb{S} = \ell_p$ with $p \in (1, \infty)$ and present an explicit description of a countable (domain-theoretic) basis for the $\omega$-continuous lattice $\mathbb{C}(\overline{\ell_p}^h)$.

Assume that $Z \subseteq \ell_p$ is a countable dense subset of
$\ell_p$. We know that $\overline{\ell_p}^h$ is a compact metric space
that embeds into the product space $\mathbb{R}^Z$ with the product
topology. Following Definition~\ref{def_pointwise}, the collection of
all sets of the form
$S(x,U) = \{f \in \mathrm{Lip}^1_{x_0}(\ell_p) \mid f(x) \in U\}$ is a
subbasis for the (relative) product topology on
$\mathrm{Lip}^1_{x_0}(\ell_p)$, in which $x$ ranges over elements of
$Z$, and $U$ ranges over open subsets of $\mathbb{R}$. Since
$\mathbb{R}$ is second-countable, we consider the collection:
\begin{equation*}
    \begin{cases}
   \mathcal{S} := \{U(d,q,\epsilon) \mid d \in Z, \ q \in \mathbb{Q}, \  \epsilon \in \mathbb{Q}^+  \},\\
   U(d,q,\epsilon) \coloneqq \{h \in \overline{\ell_p}^h  \mid | h(d) - q| < \epsilon \}.     
    \end{cases}
\end{equation*}
As such, the set $\mathcal{S}$ is countable. The set of all finite intersections of elements of $\mathcal{S}$ forms the countable basis
\begin{equation*}
	\mathcal{B} = \left \{ \bigcap U \mid U \subseteq_f \mathcal{S} \right \}    
\end{equation*}
for $\overline{\ell_p}^h$, which we enumerate as $\mathcal{B} = \{ V_i \mid i \in \Nat\}$. Recall that $U \subseteq_f \mathcal{S}$ denotes $U$ is a finite subset of $\mathcal{S}$. Finally, we obtain the following theorem:
\begin{theorem}
    The lattice $\mathbb{C}(\overline{\ell_p}^h)$, ordered by reverse inclusion, is an $\omega$-continuous lattice with a countable basis
    \begin{equation*}
      \mathcal{K} \coloneqq \left\{ \bigcup_{i \in F} \overline{V}_i \mid F \subseteq_f \Nat \right\}.  
    \end{equation*}
\end{theorem} 

\begin{proof}
Note that $\overline{\ell_p}^h$ is a second-countable compact 
Hausdorff space. As a result, a subset of $\overline{\ell_p}^h$ is closed if and only if it is compact, which entails that $\mathbb{C}(\overline{\ell_p}^h)$ is the so-called upper space of $\overline{\ell_p}^h$. Therefore, by~\cite[Proposition~3.4(i)]{Edalat95:DT-fractals}, $\mathbb{C}(\overline{\ell_p}^h)$ is an $\omega$-continuous lattice with a basis consisting of finite unions of closures of relatively compact open subsets of $\overline{\ell_p}^h$.
\end{proof}

The horofunction compactification of a given $\mathbb{S}$ can be significantly larger than $\mathbb{S}$. For example, if we take $S_{\ell_2} = \{x \in \ell_2 \mid \|x\|_2 = 1\}$ to be the unit sphere of the (infinite-dimensional) Hilbert space $\ell_2$, then its horofunction compactification $\overline{S_{\ell_2}}^h$ is homeomorphic to the closed unit ball $B_{\ell_2} = \{x \in \ell_2 \mid \|x\|_2 \leq 1 \}$ in the weak topology~\cite[Example 1.7]{daniilidis2025horofunction}. 

Nonetheless, we present an example showing that the horofunction construction can retain a higher degree of precision than the construction of~\cite{farjudian2023robustness}.
We focus on the case of $\ell_p$ spaces. For a detailed account of the metric compactification of $\ell_p$ spaces, see~\cite{gutierrez2019metric}.

\begin{example}
\label{example:ell_p_precision}
Assume that $p \in (1, \infty)$ and let $E := \{e_n \in \ell_p \mid n \in \Nat \}$, in which:
\begin{equation*}
    \forall m, n \in \Nat: \quad e_n(m) = \begin{cases}
        0 \quad n \neq m, \\
        1 \quad n = m.
    \end{cases}    
\end{equation*}
The set $E$ is a non-convex closed subset of $\ell_p$. For the basepoint $x_0 = \mathbf{0}$, we have
\begin{equation*}
\forall x \in \ell_p: \quad
    h_{{e}_n}(x) = ||x - e_n|| - \|e_n\| = \left(|x_n-1|^p + \sum_{k \neq n} |x_k|^p \right)^{1/p} -1.
\end{equation*}
Since $x \in \ell_p$, we have $\lim_{n \to \infty} x_n = 0$ and $\lim_{n \to \infty} \sum_{k\neq n}|x_k|^p = \|x\|^p$. Hence:
\begin{equation*}
  \lim_{n \to \infty} h_{e_n}(x) = \left( 1 + \|x\|^p \right)^{1/p} -1.  
\end{equation*}

Let us define $\varphi(x) \coloneqq \left( 1 + \|x\|^p \right)^{1/p} -1$. We know that $\varphi \in \overline{h(E)} \in \mathbb{C}(\overline{\ell_p}^h)$. We now show that $\forall z \in \ell_p: \varphi \neq h_z$. To obtain a contradiction, assume that $\varphi = h_z$, for some $z \in \ell_p$:

\begin{description}
    \item[Case 1, $z \neq \mathbf{0}$:] Since $x_0 = \mathbf{0}$, we have
    \begin{equation}
    \label{label:ex_eq}
     \forall x \in \ell_p: \quad
        (1 + \|x\|^p)^{1/p} - 1= \|x-z\| - \| z\|.
    \end{equation}
If we let $x = z$, then we must have:
\begin{equation}
\label{eq:1_norm_z_p}
 (1+ \|z\|^p)^{1/p} -1 = -\| z\| \implies 1 - \|z\| = (1+ \|z\|^p)^{1/p}.   
\end{equation}
But, when $z \neq \mathbf{0}$, we have $1 - \|z\| < 1 < (1+ \|z\|^p)^{1/p}$, which contradicts~\eqref{eq:1_norm_z_p}.

\item[Case 2, $z = \mathbf{0}$:] 
By inserting $z=0$ into equation (\ref{label:ex_eq}), we obtain $\forall x \in \ell_p: (1+\|x\|^p)^{1/p} = 1 + \|x\|$, which contradicts the fact that:
\begin{equation*}
\forall x \in \ell_p \setminus \{ \mathbf{0} \}: \quad (1+\|x\|^p)^{1/p} <  1 + \|x\|,
\end{equation*}
\end{description}
In particular, $\varphi \neq h_{e_n}$, for any $n \in \Nat$, and we have gained a new point. Hence, $E \cup \{\varphi\} \subseteq \overline{E}^h$, and in fact, it is straightforward to show that $E \cup \{ \varphi \} = \overline{E}^h$. Therefore, we have $\iota_*(E) = E \cup \{\varphi\}$. Since $\forall z \in \ell_p: \varphi \neq z$, we have $\iota^* \circ \iota_* (E) = E$, which shows that, over $E$, there is no loss of precision. \emph{This is in contrast with the framework of~\cite{farjudian2023robustness}}, which leads to loss of precision over $E$~\cite[Example~5.20]{farjudian2023robustness}.
\end{example}

\section{Concluding Remarks}
\label{sec:concluding_remarks}

We have proposed Gromov's horofunction compactification as a viable construction for robustness analysis over non-locally-compact metric spaces. These spaces include all infinite-dimensional Banach spaces (e.g., $\ell_p$ and $L_p(X)$ spaces). Since the embedding associated with horofunction compactification can be Lipschitz (Proposition~\ref{prop_embedding_lipschitz}), we have the foundation for computable robust approximations (Corollary~\ref{cor:compos_approx}). We presented some further analysis for the case of $\mathbb{S} = \ell_p$ spaces (Section~\ref{sec:case_ell_p}) and provided an explicit description of a countable basis for the lattice $\mathbb{C}(\overline{\ell_p}^h)$. Finally, via Example~\ref{example:ell_p_precision}, we demonstrated that the horofunction approach retains more precision compared to the construction of~\cite{farjudian2023robustness}. 

As such, our results demonstrate that the horofunction compactification offers a principled, computable, and precision-preserving foundation for robustness analysis over non-locally-compact metric spaces, including infinite-dimensional Banach spaces.

Although we have presented an explicit countable basis for $\mathbb{C}(\overline{\ell_p}^h)$, we have not investigated when, in general, the lattice $\mathbb{C}(\overline{\mathbb{S}}^h)$ can be given an effective structure~\cite{Smyth:effectively_given_domains:1977}. Furthermore, we have not studied the case of $L_p(X)$ spaces, which are ubiquitous in functional analysis and partial differential equations.

\bibliographystyle{entics}
\bibliography{Biblio}

\begin{thebibliography}{10}
\providecommand{\url}[1]{\texttt{#1}}
\providecommand{\urlprefix}{ }
\providecommand{\eprint}[2][]{\url{#2}}

\bibitem{AbramskyJung94-DT}
Abramsky, S. and A.~Jung, \emph{Domain theory}, in: S.~Abramsky, D.~M. Gabbay
  and T.~S.~E. Maibaum, editors, \emph{Handbook of Logic in Computer Science},
  volume~3, pages 1--168, Clarendon Press, Oxford (1994).\\
  \url{https://doi.org/10.1093/oso/9780198537625.003.0001
}

\bibitem{Dagnino_Farjudian_Moggi:Robustness_Quantales_Hausdorff_Smyth:ICTAC:2023}
Dagnino, F., A.~Farjudian and E.~Moggi, \emph{Robustness in metric spaces over
  continuous quantales and the {Hausdorff-Smyth} monad}, in:
  E.~{\'A}brah{\'a}m, C.~Dubslaff and S.~L.~T. Tarifa, editors,
  \emph{Theoretical Aspects of Computing -- ICTAC 2023}, volume 14446 of
  \emph{Lecture Notes in Computer Science}, pages 313--331, Springer Nature
  Switzerland, Cham (2023), ISBN 978-3-031-47963-2.
\newline\urlprefix\url{https://doi.org/10.1007/978-3-031-47963-2_19}

\bibitem{Dagnino_Farjudian_Moggi:Robust_Topology:arXiv:2025}
Dagnino, F., A.~Farjudian and E.~Moggi, \emph{Robust topology and the
  {Hausdorff-Smyth} monad on metric spaces over continuous quantales} (2025).
  \eprint{2508.11623}.
\newline\urlprefix\url{https://doi.org/10.48550/arXiv.2508.11623}

\bibitem{daniilidis2025horofunction}
Daniilidis, A., M.~Garrido, J.~Jaramillo and S.~Tapia-Garc{\'\i}a,
  \emph{Horofunction extension and metric compactifications}, Transactions of
  the American Mathematical Society, Series B \textbf{12}, pages 1130--1155
  (2025).
\newline\urlprefix\url{https://doi.org/10.1090/btran/234}

\bibitem{Edalat95:DT-fractals}
Edalat, A., \emph{Dynamical systems, measures and fractals via domain theory},
  Information and Computation \textbf{120}, pages 32--48 (1995).
\newline\urlprefix\url{https://doi.org/10.1006/inco.1995.1096}

\bibitem{Edalat_Farjudian_Li:Temporal_Discretization:2023}
Edalat, A., A.~Farjudian and Y.~Li, \emph{Recursive solution of initial value
  problems with temporal discretization}, Theoretical Computer Science page
  114221 (2023).
\newline\urlprefix\url{https://doi.org/10.1016/j.tcs.2023.114221}

\bibitem{Farjudian_Jung:Spectral:ENTICS:2024}
Farjudian, A. and A.~Jung, \emph{Continuous domains for function spaces using
  spectral compactification}, Electronic Notes in Theoretical Informatics and
  Computer Science \textbf{Volume 4 - Proceedings of MFPS XL}, 8 (2024).
\newline\urlprefix\url{https://doi.org/10.46298/entics.14736}

\bibitem{farjudian2023robustness}
Farjudian, A. and E.~Moggi, \emph{Robustness, {Scott} continuity, and
  computability}, Mathematical Structures in Computer Science \textbf{33},
  pages 536--572 (2023).
\newline\urlprefix\url{https://doi.org/10.1017/S0960129523000233}

\bibitem{gaubert2012maximin}
Gaubert, S. and G.~Vigeral, \emph{A maximin characterisation of the escape rate
  of non-expansive mappings in metrically convex spaces}, in:
  \emph{Mathematical Proceedings of the Cambridge Philosophical Society},
  volume 152, pages 341--363, Cambridge University Press (2012).
\newline\urlprefix\url{https://doi.org/10.1017/S0305004111000673}

\bibitem{Goubault-Larrecq:Non_Hausdorff_topology:2013}
Goubault-Larrecq, J., \emph{Non-Hausdorff Topology and Domain Theory},
  Cambridge University Press (2013).\\
  \url{https://doi.org/10.1017/CBO9781139524438}

\bibitem{Gromov:Horofunctions:1981}
Gromov, M., \emph{Hyperbolic manifolds, groups and actions}, in: \emph{Riemann
  surfaces and related topics: Proceedings of the 1978 Stony Brook Conference
  (State Univ. New York, Stony Brook, NY, 1978)}, volume~97, pages 183--213
  (1981).

\bibitem{gutierrez2019metric}
Guti{\'e}rrez, A.~W., \emph{On the metric compactification of
  infinite-dimensional $\ell_p$ spaces}, Canadian Mathematical Bulletin
  \textbf{62}, pages 491--507 (2019).
\newline\urlprefix\url{https://doi.org/10.4153/S0008439518000681}

\bibitem{moggi2018safe}
Moggi, E., A.~Farjudian, A.~Duracz and W.~Taha, \emph{Safe \& robust
  reachability analysis of hybrid systems}, Theoretical Computer Science
  \textbf{747}, pages 75--99 (2018).
\newline\urlprefix\url{https://doi.org/10.1016/j.tcs.2018.06.020}

\bibitem{moggi2019system}
Moggi, E., A.~Farjudian and W.~Taha, \emph{System analysis and robustness}, in:
  T.~Margaria, S.~Graf and K.~G. Larsen, editors, \emph{Models, Mindsets, Meta:
  The What, the How, and the Why Not? Essays Dedicated to Bernhard Steffen on
  the Occasion of His 60th Birthday}, pages 36--44, Springer International
  Publishing (2019), ISBN 978-3-030-22348-9.
\newline\urlprefix\url{https://doi.org/10.1007/978-3-030-22348-9_4}

\bibitem{Moggi_Farjudian_Taha:System_Analysis_and_Robustness:ICTCS:2019}
Moggi, E., A.~Farjudian and W.~Taha, \emph{System analysis and robustness}, in:
  \emph{Proceedings of the 20th Italian Conference on Theoretical Computer
  Science, {ICTCS} 2019, Como, Italy, September 9-11, 2019}, pages 1--7 (2019).
\newline\urlprefix\url{http://ceur-ws.org/Vol-2504/paper1.pdf}

\bibitem{munkres2017topology}
Munkres, J., \emph{Topology}, Pearson Modern Classics for Advanced Mathematics
  Series, Pearson (2017), ISBN 9780134689517.
\newline\urlprefix\url{https://books.google.co.uk/books?id=51n8MAAACAAJ}

\bibitem{Smyth:effectively_given_domains:1977}
Smyth, M.~B., \emph{Effectively given domains}, Theoretical Computer Science
  \textbf{5}, pages 257--274 (1977).
\newline\urlprefix\url{https://doi.org/10.1016/0304-3975(77)90045-7}

\bibitem{Zhou_Shaikh_Li_Farjudian:Robust_NN:MSCS:2023}
Zhou, C., R.~A. Shaikh, Y.~Li and A.~Farjudian, \emph{A domain-theoretic
  framework for robustness analysis of neural networks}, Mathematical
  Structures in Computer Science \textbf{33}, pages 68--105 (2023).
\newline\urlprefix\url{https://doi.org/10.1017/S0960129523000142}

\end{thebibliography}

\end{document}